\documentclass[11pt,a4paper]{article}
\usepackage{fullpage}
\usepackage{xspace}
\usepackage{amsmath,upgreek}
\usepackage{amsfonts}
\usepackage{amssymb}
\newtheorem{theorem}{Theorem}

\newtheorem{lemma}[theorem]{Lemma}

\newcommand{\Xomit}[1]{ }
\newenvironment{proof}[1][Proof]{\textbf{#1.} }{\ \rule{0.5em}{0.5em}}

\newcommand{\eps}{\upvarepsilon}

\begin{document}

\title{Lower bounds for several online variants of bin packing\footnote{Gy. D\'{o}sa was supported by VKSZ\_12-1-2013-0088 ``Development of cloud based
smart IT solutions by IBM Hungary in cooperation with the
University of Pannonia'' and by National Research, Development and
Innovation Office -- NKFIH under the grant SNN 116095. L. Epstein
and A. Levin were partially supported by a grant from GIF - the
German-Israeli Foundation for Scientific Research and Development
(grant number I-1366-407.6/2016).}}

\date{}

\author{J\'anos Balogh \thanks{Department of Applied Informatics, Gyula Juh\'asz Faculty of Education,
     University of Szeged, Hungary. \texttt{balogh@jgypk.u-szeged.hu}} \and J\'ozsef B\'ek\'esi \thanks{Department of Applied Informatics, Gyula Juh\'asz Faculty of Education,
     University of Szeged, Hungary. \texttt{bekesi@jgypk.u-szeged.hu}} \and Gy\"{o}rgy
D\'{o}sa\thanks{Department of Mathematics, University of Pannonia,
Veszprem, Hungary, \texttt{dosagy@almos.vein.hu}.} \and Leah
Epstein\thanks{ Department of Mathematics, University of Haifa,
Haifa, Israel. \texttt{lea@math.haifa.ac.il}. } \and Asaf
Levin\thanks{Faculty of Industrial Engineering and Management, The
Technion, Haifa, Israel. \texttt{levinas@ie.technion.ac.il.}}}

\vspace{-0.5cm}

\maketitle

\begin{abstract}
We consider several previously studied online variants of bin
packing and prove new and improved lower bounds on the asymptotic
competitive ratios for them. For that, we use a method of fully
adaptive constructions. In particular, we improve the lower bound
for the asymptotic competitive ratio of online square packing
significantly, raising it from roughly $1.68$ to above $1.75$.
\end{abstract}

\section{Introduction}\label{sec:intro}

In bin packing problems, there is an input consisting of a set of
items, and the goal is to partition it into a minimum number of
subsets called bins, under certain conditions and constraints. In
the classic variant \cite{Ullman71,J74,JDUGG74,Yao80A,Seiden02J},
items have one-dimensional rational numbers in $(0,1]$, called
sizes, associated with them, and the total size of items of one
bin cannot exceed $1$. In online variants items are presented as a
sequence and the partition is created throughout this process in
the sense that any new item should be assigned to a bin before any
information regarding the next item is provided. The conditions on
the partition or packing remain as in the offline problem where
the items are all given at once as a set. Using an algorithm $A$
to partition the items into subsets, which is also seen as a
process of packing items into bins, the number of partitions or
bins used for the packing is defined to be the cost of $A$.

Algorithms for bin packing problems are normally studied using the
asymptotic approximation ratio, also called asymptotic competitive
ratio for the case of online algorithms (and we will use this last
term). For an algorithm $A$ and an input $I$, let $A(I)$ denote
the number of bins used by $A$ for $I$, that is, the cost of $A$
for $I$. Let $OPT(I)$ denote the number of bins that an optimal
solution uses for $I$, that is, the cost of an optimal (offline)
algorithm $OPT$ for $I$. Consider the set of inputs $J_Q$ of all
inputs for which the number of bins used by $OPT$ is $Q$. For the
problems studied here (and non-empty inputs for them), $Q$ will be
a positive integer. Let $c(Q)=\max_{I \in J_Q} A(I)$ (where for
reasonable algorithms this value is finite), and let $R_A =
\limsup_{Q \rightarrow \infty} \frac{c(Q)}Q$. The absolute
competitive ratio of $A$ is defined by $\sup_{I}
\frac{A(I)}{OPT(I)}$, that is, this is the supremum ratio between
the cost of $A$ and the optimal cost, over all inputs, and the
asymptotic competitive ratio is the superior limit of the absolute
competitive ratios for fixed values of $Q=OPT(I)$ when $Q$ grows
to infinity. Since the standard measures for online bin packing
problems (and offline bin packing problems, respectively), are the
asymptotic competitive ratio (and the asymptotic approximation
ratio), we also use the terms {\it competitive ratio} (and {\it
approximation ratio}) for them, and always use the word {\it
absolute} when we discuss the absolute measures. To prove lower
bounds on the (asymptotic) competitive ratio one can use inputs
where the optimal cost is arbitrarily large, and we use this
method.  The study of lower bounds on the competitive ratio for a
given problem characterizes the extent to which the performance of
the system deteriorates due to lack of information regarding the
future input items.

Here, we study three versions of the online bin packing problem,
providing new lower bounds on the competitive ratio for them.
Previous constructions used for proving such lower bounds were
often inputs where items arrive in batches, such that the items of
one batch all have the exact same size (and the input may stop
after a certain batch or it can continue to another one). In the
known lower bounds for classic bin packing, it is even known what
the next batches will be, if they are indeed presented
\cite{Liang80,Vliet92,BBG}. While it may be obvious that adaptive
inputs where the properties of the next item are based on the
packing of previous items are harder for an algorithm to deal
with, it was not known until recently how to use this idea for
designing lower bounds, except for special cases
\cite{Blitz,BCKK04,FK13}. In cardinality constrained bin packing
\cite{KSS75,KP99,Epstein05,BCKK04,BDE}, items are one-dimensional,
a fixed integer $t\geq 2$ is given, and the two requirements for a
packed bin are that its total size of items is at most $1$, and
that it contains at most $t$ items. The special case analyzed in
the past \cite{Blitz,BCKK04,FK13} is $t=2$, which can also be seen
as a matching problem, as every bin can contain at most two items.
In \cite{BBDEL} we showed that the overall competitive ratio
(supremum over all values of $t$) is $2$ (an upper bound was known
prior to that work \cite{BCKK04,BDE}), and provided improved lower
bounds for relatively small values of $t$. For standard bin
packing, the best known lower bound on the competitive ratio is
$1.5403$ \cite{Vliet92,BBG} and the best upper bound is $1.57829$
\cite{BBDEL_ub}.

Another lower bound presented in \cite{BBDEL} is for the
competitive ratio of vector packing in at least two dimensions.
For an integer dimension $d \geq 2$, the items have
$d$-dimensional vectors  associated with them, whose components
are rational numbers in $[0,1]$ (none of which are all-zero
vectors), and bins are all-one vectors of dimension $d$. A subset
of items can be packed into a bin if taking no component exceeds
$1$ in their vector sum. This generalizes cardinality constrained
bin packing, and we showed a lower bound of $2.03731129$ on the
competitive ratio of the online variant for any $d \geq 2$ (prior
to that work, no lower bound strictly above $2$ for a constant
dimension was known).

Our main goal here is to exhibit how to exploit adaptive
constructions with some connection to those used in \cite{BBDEL}
in order to obtain lower bounds for other variants. We focus on
the following three variants.  In all three variants of online bin
packing which we study, the input consists of rational numbers in
$(0,1]$, however there is additional information received with the
input in some of the cases and the input is interpreted in
different ways. Two of the problems are one-dimensional and the
input numbers are sizes of items. The third variant is
two-dimensional, and the numbers are side lengths of squares. In
our first variant called {\it bin packing with known optimal
cost}, the cost of an optimal (offline) solution is given in
advance, that is, it is known how many bins are required for
packing the input. This problem is also called {\sc K-O}
(known-OPT). It is currently hard to find an appropriate way to
use this additional piece of information for algorithm design, but
in all lower bounds known for standard online bin packing
\cite{Vliet92,BBG} the property that the optimal cost is different
for different inputs is crucial for achieving the result. For {\sc
K-O}, a lower bound of 1.30556 on the competitive ratio was
presented \cite{EL08x} and later improved to 1.32312 \cite{BB13}.
We show a new lower bound of $\frac {87}{62}\approx 1.4032258$ on
the competitive ratio, improving the previous result
significantly. This problem is related to the field of semi-online
algorithms and to the so-called model of {\it online algorithms
with advice} \cite{BKLLO,ADKRR}, where the online algorithm is
provided with some (preferably very small) pieces of information
regarding the input.

In the square packing ({\sc SP}) problem, the goal is to assign an
input set of squares whose sides are rational numbers in $(0,1]$
into bins that are unit squares in a non-overlapping and
axis-parallel way, so as to minimize the number of non-empty bins.
We use the standard definition of this packing problem, where two
squares do not overlap if their interiors do not overlap (but they
may have common points on the boundaries of the squares). The
offline variant is well-studied \cite{BCKS06,EL13}. The history of
lower bounds on the competitive ratio of online algorithms for
this problem is as follows. Several such lower bounds were proved
for the online version of {\sc SP}, starting with a simple
construction yielding a lower bound of $\frac 43$ on the
competitive ratio by Coppersmith and Raghavan \cite{CopRag89}, and
then there were several improvements
\cite{SeiSte03,EpsSte05,HvScube}, all showing bounds above $1.6$.
In 2016 a copy of the thesis of Blitz \cite{Blitz} from 1996 was
found by the authors of \cite{HvScube}. This thesis contains a
number of lower bounds for bin packing problems, including a lower
bound of 1.680783 on the competitive ratio of online algorithms
for {\sc SP}. The result of Blitz \cite{Blitz} is now the previous
best lower bound on the competitive ratio for the problem (prior
to our work), and it is higher than the lower bounds of
\cite{SeiSte03,EpsSte05,HvScube}. Here, we show a much higher
lower bound, larger than $1.7515445$, on the competitive ratio of
this problem.

Finally, we consider class constrained bin packing ({\sc CLCBP})
\cite{ST01,ST04,XM06,EIL10}. In this one-dimensional variant every
item has a size and a color, and for a given parameter $t \geq 1$,
any bin can receive items of at most $t$ different colors (of
total size at most $1$), while the number of items of each color
can be arbitrary. This problem generalizes standard bin packing,
as for any input of standard bin packing, defining a common color
to all items results in an instance of {\sc CLCBP} for any $t$. It
also generalizes bin packing with cardinality constraints, though
here to obtain an instance of {\sc CLCBP} one should assign
distinct colors to all items. We provide improved lower bounds for
$t=2,3$. For $t=2$, the previous known lower bound was $1.5652$
\cite{EIL10}. For $t=3$, the previous lower bound was $\frac 53
\approx 1.6667$ \cite{ST04}. This last result was proved even for
the special case with equal size items. Interestingly, it has
elements of adaptivity, but with respect to colors (as all items
have identical sizes), and the input moves to presenting items of
a new color once the algorithm performs a certain action. We show
that the competitive ratio of any online algorithm for {\sc CLCBP}
with $t=2$ is at least $1.717668$, and that the competitive ratio
of any online algorithm for {\sc CLCBP} with $t=3$ is at least
$1.808142$.

The drawback of previous results for all those problems is that
while the exact input was not known in advance, the set of sizes
used for it was determined prior to the action of the algorithm.
We show here that our methods for proving lower bounds can be
combined with a number of other approaches to result in improved
lower bounds for a variety of bin packing problems. We use the
following theorem proved in \cite{BBDEL} (see the construction in
Section 3.1 and Corollary 3).

\begin{theorem}\label{cnstrct}

(i) Let $N \geq 1$ and $k \geq 2$ be large positive integers.
Assume that we are given an arbitrary deterministic online
algorithm for a variant of bin packing and a condition $C_1$ on
the possible behavior of an online algorithm for one item (on the
way that the item is packed). An adversary is able to construct a
sequence of values $a_i$ ($1 \leq i \leq N$) such that for any
$i$, $a_i \in \left( k^{- 2^{N+3}}, k^{-2^{N+2}} \right)$, and in
particular $a_i \in \left(0,\frac{1}{k^4}\right)$. For any item
$i_1$ satisfying $C_1$ and any item $i_2$ not satisfying $C_1$, it
holds that $\frac{a_{i_2}}{a_{i_1}} > k$. Specifically, there are
values $\beta$ and $\gamma$ such that for any item $i_1$
satisfying $C_1$,
 and any item $i_2$ not satisfying $C_1$, it holds
that $a_{i_1}<\gamma <a_{i_2}$ and $\frac{a_{i_2}}{a_{i_1}}>\beta$.

(ii) If another condition $C'$ is given for stopping the input (it
can be a condition on the packing or on the constructed input), it
is possible to construct a sequence $a_i$ consisting of $N$ items
such that $C'$ never holds, or a sequence of $N'<N$ items, such
that $C'$ holds after $N'$ items were introduced (but not
earlier), and where the sequence satisfies the requirements above.
\end{theorem}

Examples for the condition $C_1$ can be the following: ``the item
is packed as a second item of its bin'', ``the item is packed into
a non-empty bin'', ``the item is packed into a bin an item of size
above $\frac 12$'', etc. An example for the condition $C'$ can be
``the algorithm has at least a given number of non-empty bins''.

The construction of such inputs is based on presenting items one
by one, where there is an active (open) interval of sizes out of
which future values $a_i$ are selected. When a new item is
presented, and the algorithm packs it such that it does not
satisfy $C_1$, all future items will be smaller. If the algorithm
packs a new item such that it satisfies $C_1$, all future items
will be larger. This reduces the length of the active interval.
Thus, even though the active interval becomes shorter in every
step where a new item arrives, it always has a positive length.
One can see this as a kind of binary search on the value $\gamma$,
which will always be contained in the remaining interval (as it
remains non-empty). For example, Fujiwara and Kobayashi
\cite{FK13} used a similar approach and in their work the middle
point of the active interval is the size of the next item, and the
active interval has length that it smaller by a factor of $2$
after every step. To obtain the stronger property that items whose
sizes is at least the right endpoint of the active interval are
larger by a factor of $k$ than items no larger than the left
endpoint of the active interval, the selection of the next size is
performed by a process similar to geometrical binary search.

Note that an important feature is that the value $a_i$ is defined
{\it before} it is known whether $C_1$ holds for the $i$th item
(the item corresponding to $a_i$, that is, the item whose size is
a function of $a_i$).  We will use this theorem throughout the
paper. We study the problems in the order they were defined.

\section{Online bin packing with known optimal cost ({\sc K-O})}\label{sec:knownopt}
Here, we consider the problem {\sc K-O}, and prove a new lower
bound on the competitive ratio for it. We prove the following
theorem.

\begin{theorem}
The competitive ratio of any online algorithm for {\sc K-O} is at
least  $\frac {87}{62}\approx 1.4032258$.
\end{theorem}

Let $M$ be a large integer that is divisible by $4$ ($M$ will be
the value of the known optimal cost). We will create several
alternative inputs, such that the optimal cost will be equal to
$M$ for each one of them.

We use the following construction.  For $k=10$ and $N=M$, define
an input built using Theorem \ref{cnstrct} as follows applied
twice on different parts of the input as explained below. The
outline of our lower bound construction is as follows. The first
part of the input will consist of $M$ items of sizes slightly
above $\frac 17$ (such that some of them, those packed first into
bins, are larger than the others). Then, there are $M$ items of
sizes slightly above $\frac 13$ (where items packed into new bins
are larger than others, while those combined with items of sizes
roughly $\frac 17$ or with another item of size roughly $\frac
13$, or both, are slightly smaller). Finally, the algorithm will
be presented with a list of identical items of one of the three
sizes $1$ (exactly), or slightly above $\frac 12$, or slightly
below $\frac 23$, such that every larger item of size slightly
above $\frac 13$ cannot be packed together with such an item (of
size slightly below $\frac 23$). Additionally, after the first $M$
items arrive, it is possible that instead of the input explained
here there are items of sizes slightly below $\frac 67$, either
such that every such item can be packed with any item out of the
first $M$ items, or such that it can only be combined with the
smaller items out of the first $M$ items (due to the property that
the size of an item will be just below $\frac 67$, in both cases
it can be combined with at most one item of size just above $\frac
17$).

Next, we formally define our input sequences. Throughout this
section, let the condition $C_1$ be that the item is not packed as
a first item into a bin. The first $M$ items are defined as
follows. Using Theorem \ref{cnstrct}, we create $M$ items such
that the size of item $i$ is $\frac 17+a_i$. These items are
called $S$-items. The sizes of such items are in $(\frac 17,
0.143)$, and there is a value $\gamma_1$ such that any item whose
packing satisfies condition $C_1$ has size below $\frac 17 +
\gamma_1$ and any item whose packing does not satisfy $C_1$ has
size above $\frac 17 + \gamma_1$. The first kind of items are
called small $S$-items, and the second kind of items are called
large $S$-items.

Let $Y_7$ denote the current number of bins used by the algorithm
(after all $S$-items have arrived), and this is also the number of
large $S$-items. Two possible continuations at this point are $M$
items of sizes equal to $\frac 45$ (the first option), and $M -
\lceil \frac{Y_7}{6} \rceil$ items of sizes equal to $\frac 67 -
\gamma_1$ (the second option).

\begin{lemma}\label{M1}
In both options, an optimal solution has cost $M$.
\end{lemma}
\begin{proof}
In the first option, an optimal solution has one item of size
$\frac 45$ and one item of size no larger than $0.143$ in every
bin. It is optimal as every item of size above $\frac 12$ requires
a separate bin (where it can be possibly packed with smaller
items).

In the second option, an optimal solution uses $\lceil
\frac{Y_7}{6} \rceil$ bins to pack the large $S$-items: Every bin
can contain at most six such items, as their sizes are in $(\frac
17,\frac 16)$, each remaining bin has one item of size $\frac 67 -
\gamma_1
> 0.857$, and $M-Y_7$ of them also have one item (each) of size
below $\frac 17 +\gamma_1$. This is an optimal solution as the two
larger kinds of items (those of sizes above $\frac 12$ and the
 large $S$-items) cannot be combined into the same bins, and the packing
for each of these two kinds of items is optimal.
\end{proof}

In the first case, the algorithm can use bins containing exactly
one item to pack (also) an item of size $\frac 45$, but it cannot
use any other bin again. In the second case, as every bin has
exactly one item of size above $\frac 17+\gamma_1$, the algorithm
uses an empty bin for every item of size $\frac 67 - \gamma_1$.

We explain the continuation of the input in the case where none of
the two continuations already defined is used. The next $M$ items
are defined using Theorem \ref{cnstrct}, and we create $M$ items
such that the size of the $i$th item of the current subsequence of
$M$ items is $\frac 13+a_i$ (the values $a_i$ are constructed here
again, and they are different from the values $a_i$ constructed
earlier). We call these items $T$-items. The sizes of $T$-items
are in $(\frac 13, 0.33344)$, and there is a value $\gamma_2$ such
that any item whose packing satisfies condition $C_1$ (defined in
this section) has size below $\frac 13 + \gamma_2$ and for any
item whose packing does not satisfy $C_1$, it has size above
$\frac 13 + \gamma_2$. The first kind of items are called small
$T$-items, and the second type items are called large $T$-items.

Here, there are three possible continuations. The first one is
$\frac M2$ items, all of size $1$. The second one is $M$ items,
each of size $0.52$. Let $Y_3$ denote the number of new bins
created for the $T$-items, which is also the number of large
$T$-items (so after the $T$-items are packed the algorithm uses
$Y_7+Y_3$ bins). If $Y_3 \leq \frac M2$, the third continuation is
with $\frac {3M}4$ items, each of size $\frac 23 - \gamma_2$
(where $\frac 23 - \gamma_2 > 0.66656$). Otherwise ($Y_3 > \frac
M2$), the third continuation is with $M- \lceil\frac{Y_3}2 \rceil$
items, each of size $\frac 23 - \gamma_2$. Thus, in the third
continuation, the sizes of items are the same (i.e., $\frac 23
-\gamma_2$) in both cases, and the number of items is $M-
\max\{\frac M4,\lceil\frac{Y_3}2 \rceil\}$.

\begin{lemma}\label{M2}
The optimal cost in all cases (i.e., after the packing of the
items of each possible continuation has been completed) is exactly
$M$.
\end{lemma}
Note that it is sufficient to show that the optimal cost is at
most $M$, as in the case where it is strictly smaller than $M$, it
is possible to present items of size $1$ until the optimal cost is
exactly $M$, while the cost of the algorithm does not decrease. We
prove that the value is exactly $M$ to stress the property that
one cannot prove a better lower bound using the same kind of
input.

\medskip

\noindent\begin{proof} For the first continuation, an optimal
solution packs $\frac M2$ bins, each with two $S$-items and two
$T$-items, and another $\frac M2$ bins, each with one item of size
$1$. This solution is optimal as every item of size $1$ has to be
packed alone into a bin, and no bin can contain more than two
items of sizes above $\frac 13$.

For the second continuation, an optimal solution packs $M$ bins,
each with one item of size $0.52$, one $T$-item and one $S$-item.
This solution is optimal as no bin can contain more than one item
of size above $\frac 12$.

For the third continuation, the two options for optimal solutions
are as follows. In the case $Y_3 \leq \frac M2$, there are
$\frac{M}4$ bins, each with two $T$-items and two $S$-items. All
large $T$-items will be packed into these bins (which is possible
as there are $\frac M2 \geq Y_3$ $T$-items packed into those
bins). There are also $\frac M2 -Y_3$ small $T$-items packed into
these bins. Each of the remaining bins contains one item of size
$\frac 23 - \gamma_2$, where $\frac M4$ of those bins also contain
two $S$-items (which is possible as the total size will be below
$0.143\cdot 2+ \frac 23<1$), and each of the remaining $\frac M2$
bins has one small $T$-item (this is possible as the size of each
small $T$-item is below $\frac 13 + \gamma_2$).

In the case $Y_3 > \frac M2$, there are $\lceil \frac{Y_3}2
\rceil$ bins with two $S$-items and two large $T$-items (at most
one bin may contain a smaller number of large $T$-items). All
large $T$-items are packed into these bins, and no small $T$-items
are packed into these bins. The remaining bins all contain one
item (each) of size $\frac 23 - \gamma_2$, where $\frac M2 -
\lceil \frac{Y_3}2 \rceil$ of those bins also contain two
$S$-items, and $M- Y_3$ of those bins (not containing $S$-items)
also contain one small $T$-item (this is possible as $\lceil
\frac{Y_3}2 \rceil+\frac M2 - \lceil \frac{Y_3}2 \rceil +M-Y_3
\leq M$).

The solution for the second case (i.e., for the case $Y_3 > \frac
M2$) is optimal as separate bins are needed for items of size
$\frac 23 - \gamma_2$ and large $T$-items, and the solution
obtained for each kind is optimal.

Thus, it remains to prove that in the first case (i.e., in the
case $Y_3 \leq \frac M2$), the optimal cost is $M$. Observe that
we showed a feasible solution of cost $M$, so we need to show that
the optimal cost is at least $M$. In this case every bin with an
item of size $\frac 23 - \gamma_2$ can receive either two
$S$-items or one small $T$-item. Consider an optimal solution and
let $\Delta \geq 0$ be the number of items of size $\frac 23 -
\gamma_2$ packed with a $T$-item. The remaining $(M-\Delta)$
$T$-items are packed at most two in each bin, so if $\Delta \leq
\frac M2$, we are done as there are at least $\frac{3M}4 + \frac
{M-\Delta}2 \geq M$ bins. Otherwise, $\Delta \geq \frac{M}2+1$, at
most $2(\frac{3M}4-\Delta)$ $S$-items
 are packed with items of size $\frac 23 - \gamma_2$, and
$M-2(\frac{3M}4-\Delta)=2\Delta-\frac M2$ $S$-items remain to be
packed with $(M-\Delta)$ $T$-items. Even replacing each $T$-item
with two items of size in $(\frac 17,\frac 16]$ (virtually, for
the sake of proof), we have to pack $2(M-\Delta)+2\Delta-\frac
M2=\frac{3M}2$ items where a bin can contain at most six items, so
at least $\frac M4$ bins are needed, for a total of
$\frac{3M}4+\frac M4=M$ bins.
\end{proof}

This completes the description of the input where we showed that
in each case the optimal cost is exactly $M$.  Next, we consider
the behavior of the algorithm. Consider the kinds of bins the
algorithm may have after all $T$-items have arrived. The $T$-items
do not necessarily arrive, but we will deduce the numbers of
different kinds of bins the algorithm has after the $S$-items have
arrived from the numbers of bins assuming that the $T$-items have
arrived. This is an approach similar to that used in
\cite{Vliet92}, where numbers of bins packed according to certain
patterns (subsets of items that can be packed into one bin) at the
end of the input are considered, and based on them, the number of
bins already opened at each step of the input are counted. More
precisely, if the input consists of batches of identical (or
similar) items, given the contents of a bin it is clear when it is
opened and at what times (after arrival of sub-inputs) it should
be counted towards the cost of the algorithm.

A bin with no $T$-items can receive an item of size $0.52$ if it
has at most three $S$-items and it can receive an item of size
$\frac 23 - \gamma_2$ if it has at most two $S$-items. The only
case where a bin with at least one $S$-item and at least one
$T$-item can receive another item (out of a continuation of the
input) is the case that a bin has one of each of these types of
items, and it will receive an item of size $0.52$.

Let $X_{60}$ denote the number of bins with four or five or six
$S$-items and no $T$-items. Such a bin cannot receive any further
items in addition to its $S$-items. Let $X_{30}$ denote the number
of bins with three $S$-items and no $T$-items. Such a bin can
receive an item of size $0.52$ (but not a larger item). Let
$X_{20}$ and $X_{10}$ denote the number of bins with two $S$-items
and one $S$-item, respectively, and no $T$-items. Out of possible
input items, such a bin can receive an item of size $0.52$ or an
item of size $\frac 23- \gamma_2$. We distinguish these two kinds
of bins due to the possible other continuations after $T$-items
have arrived. Let $X_{41}$ denote the number of bins with two or
three or four $S$-items and one $T$-item. Such bins cannot receive
any further items out of our inputs. Let $X_{11}$ denote the
number of bins with one $S$-item and one $T$-item. Let $X_{12}$
and $X_{22}$ denote the numbers of bins with two $T$-items and one
and two $S$-items, respectively. Obviously, there can be bins
without $S$-items containing one or two $T$-items, and we denote
their numbers by $X_{01}$ (one $T$-item) and $X_{02}$ (two
$T$-items).

We have five scenarios based on the different options and
continuations described above, and we use $ALG_i$ to denote the
cost of a given algorithm for each one of them, in the order they
were presented. Let $R$ be the (asymptotic) competitive ratio. Let
$A_i = \limsup _{M \rightarrow \infty} \frac{ALG_i}{M}$, which is
a lower bound on the competitive ratio $R$ since the optimal cost
is always $M$ (by Lemmas \ref{M1} and \ref{M2}), so for
$i=1,2,3,4,5$ we have the constraint $A_i \leq R$. The $A_i$ (for
$i=1,2,3,4,5$) will not appear explicitly as variables in the
forthcoming linear program. Instead, we will compute each $A_i$
based on the other variables in the program and substitute the
resulting expression in the constraint $A_i \leq R$. We use
$y_i=\frac{Y_i}{M}$ and $x_{ij}=\frac{X_{ij}}M$ for those values
of $i$ and $j$ such that $Y_i$ and $X_{ij}$ are defined. For all
thirteen variables there is a non-negativity constraint. In
addition, the number of items should satisfy $\sum_{i,j} j \cdot
X_{ij} = M$ and $\sum_{i,j} i \cdot X_{ij} \geq  M$ (the second
constraint is not an equality as in some cases $X_{ij}$ counts
bins with at most $(i)$ $S$-items). Using the definitions of $Y_7$
and $Y_3$ we have $Y_7=X_{60} + X_{30} +X_{20} + X_{10} + X_{41} +
X_{11} + X_{12} + X_{22}$ and $Y_3=X_{01}+X_{02}$.

We get the following four constraints:
\begin{eqnarray}
x_{41}  + x_{11}  + 2 x_{12} + 2 x_{22} + x_{01} + 2 x_{02} = 1  \label{eq1}\\
6 x_{60} + 3 x_{30} + 2 x_{20} + x_{10} + 4 x_{41} + x_{11} +
x_{12} + 2 x_{22} \geq 1
\label{eq2} \\
y_7-x_{60} - x_{30} -x_{20} - x_{10} - x_{41} - x_{11} - x_{12} -
x_{22} = 0 \label{eq3} \\
y_3 - x_{01} - x_{02} = 0 \label{eq4}
\end{eqnarray}

The costs of the algorithm are as follows. We have $ALG_1=M+X_{60}
+ X_{30} +X_{20} + X_{41} + X_{22}$, $ALG_2= M - \lceil
\frac{Y_7}{6} \rceil + Y_7$, $ALG_3=Y_7+Y_3+\frac M2$, and
$ALG_4=X_{60}+X_{41}+X_{22}+X_{12}+X_{02}+ M$.

If $Y_3 \leq \frac M2$, we have $ALG_5= Y_7+Y_3-X_{20}-X_{10}+
\frac{3M}4$, and if $Y_3 > \frac M2$, we have $ALG_5=
Y_7+Y_3-X_{20}-X_{10}+ M- \lceil\frac{Y_3}2 \rceil$.

The four first costs of the algorithm (for the first four
scenarios) gives the constraints

\begin{eqnarray}
R - x_{60} - x_{30} - x_{20} - x_{41} - x_{22} \geq 1 \label{eq5} \\
6 R - 5 y_7  \geq 6 \label{eq6} \\
2 R - 2 y_7 - 2 y_3 \geq 1 \label{eq7} \\
R - x_{60} - x_{41} - x_{22} - x_{12} - x_{02} \geq 1 \label{eq8}
\end{eqnarray}

The two final constraints form two cases (according to the value
of $y_3$), and therefore our list of constraints results in two
linear programs (with all previous constraints and two additional
ones). The inputs for the two cases are different, and therefore
they are considered separately (due to the different inputs, there
is one other different constraint except for the constraint on the
value of $y_3$). For each one of the linear programs, the
objective is to minimize the value of $R$.

One pair of constraints is $y_3 \leq \frac 12$ and $4 R - 4 y_7 -
4 y_3 + 4 x_{20}+ 4 x_{10} \geq 3$, and the alternative pair is
$y_3 \geq \frac 12$ and $ 2 R -  2 y_7 - y_3 + 2 x_{20} + 2 x_{10}
\geq 2$ (observe that the constraint $y_3\geq \frac 12$ is a relaxation of the valid constraint $y_3> \frac 12$, and thus the weaker constraint $y_3 \geq \frac 12$ is valid in this case).

Multiplying the first five constraints by the values $2$, $1$,
$3$, $2$, $1$, respectively, and taking the sum gives:
\begin{equation}2 x_{60}+ 2 x_{41} + 2 x_{12}   + 2 x_{02}  + 2 x_{22}
- 2 x_{10}
  - x_{30} - 2 x_{20} + 3 y_7 + 2 y_3 + R \geq 4 \ . \label{atotal}\end{equation}

For the first case, we take the sum of the sixth, eighth, and
tenth constraints multiplied by the values $2$, $20$, $5$,
respectively, and get:
$$ 52 R - 30 y_7  - 20 y_3   - 20 x_{60} - 20 x_{41} - 20 x_{22} - 20
x_{12} - 20 x_{02} + 20 x_{20}+ 20 x_{10}
 \geq 47 \ . $$ Summing this with ten times (\ref{atotal}) we get
$ 62 R   - 10 x_{30}    \geq 87 $, and by $x_{30} \geq 0$ we get
$R \geq \frac {87}{62}\approx 1.4032258$.

For the second case, we take the sum of the seventh, eighth, and
tenth constraints multiplied by the values $1$, $4$, $2$,
respectively, and get:
$$
10 R  - 6 y_7 - 4 y_3  - 4 x_{60} - 4 x_{41} - 4 x_{22} - 4 x_{12}
- 4 x_{02} + 4 x_{20} + 4 x_{10}
 \geq 9
\ . $$ Summing this with twice (\ref{atotal}) we get $12 R
  - 2 x_{30}  \geq 17$, and as $x_{30} \geq 0$, we have $R \geq \frac {17}{12} \approx 1.41666$.
Thus, we have proved $R \geq 1.4032258$.

\section{Online Square packing ({\sc SP})}\label{sec:squares}
We continue with the online square packing ({\sc SP}) problem. We
prove the following theorem.

\begin{theorem}
The competitive ratio of any online algorithm for {\sc SP} is at
least $1.7515445$.
\end{theorem}

Here, in the description of the input, when we refer to the size
of an item, this means the length of the side of the square (and
not its area). Consider the following input. For a large positive
even integer $M$ and $k=10$, we define an input based on using
Theorem \ref{cnstrct} twice. The construction is similar to that
of the previous section, though here we are not committed to a
specific optimal cost, and we take into account the
multidimensionality. Moreover, for one of the item types the
number of such items is also determined by the action of the
algorithm (which was difficult to implement in the previous
section when the cost of an optimal packing is fixed in advance,
and we did not use such an approach there as extensively as in the
current section). Here, we only compute upper bounds on the
optimal cost for each case.

The outline of the construction is as follows.  The first part of
the input will consist of $M$ items of sizes slightly above $\frac
14$ (such that some of them, those packed first into bins, are
larger than the others), then, there are items of sizes slightly
above $\frac 13$ (where such items that are packed into bins
containing relatively few items, where the exact condition is
defined below, will be larger than other items of this last kind).
Finally, there will be items of one of the sizes: $\frac 35$, and
slightly below $\frac 23$ (all of them will have exactly the same
size), such that every larger item of size slightly above $\frac
13$ cannot be packed together with such an item of size slightly
smaller than $\frac 23$. Additionally, after the first $M$ items
arrive, it is possible that instead of the input explained here
there are items of sizes slightly below $\frac 34$, such that it
can be only be combined with the smaller items out of the first
$M$ items (any bin with an item of size slightly below $\frac 34$
may have at most five smaller items out of the first $M$ items in
a common bin).

Next, we formally define the construction.  Let the condition $C_{11}$ be that the item is not packed as a
first item into a bin. This is the condition we will use for items
of sizes slightly above $\frac 14$. For items of sizes slightly
above $\frac 13$, let the condition $C_{12}$ be that the item is
either packed in a bin already containing an item of size above
$\frac 13$, or that it contains at least five items whose sizes
are in $(\frac 14,\frac 13]$.

The first $M$ items are defined as follows.  Using Theorem
\ref{cnstrct}, we create $M$ items such that the size of item $i$
is $\frac 14+a_i$. These items are called $F$-items. The sizes of
items are in $(0.25, 0.2501)$, and there is a value $\gamma_1$
such that any item whose packing satisfies condition $C_{11}$ has
size below $\frac 14 + \gamma_1$ and for any item whose packing
does not satisfy $C_{11}$, it has size above $\frac 14 +
\gamma_1$. The first kind of items are called small $F$-items, and
the second type items are called large $F$-items. No matter how
the input continues, as any packing of the first $M$ items
requires at least $ \frac{M}9$ bins, the cost of an optimal
solution is $\Omega(M)$.

Let $Y_4$ denote the current number of bins used by the algorithm,
and this is also the number of large $F$-items. A possible
continuation at this point is $\lceil \frac{M-Y_4}5 \rceil$ items
of (identical) sizes equal to $\frac 34 - \gamma_1$. Note that
such an item cannot be packed into a bin with an item of size
above $\frac 14+\gamma_1$, as it cannot be packed next to it or
below (or above) it, and the remaining space (not next to it or
below it or above it) is too small (the sum of the diagonals of
these two items is too large to be packed into a unit square bin).

\begin{lemma}\label{lema6}
There exists a packing of the items of the presented sequence (in
this case) of cost at most $\frac{M}5 - \frac{4Y_4}{45} +2$.
\end{lemma}
\begin{proof} A possible packing of the items of sizes $\frac 34-\gamma_1$
together with the $(M)$ $F$-items is to use  $\lceil \frac{M-Y_4}5
\rceil$ bins for the new items and combine five small $F$-items
into these bins (one such bin may have a smaller number of
$F$-items). This packing is feasible as the large item can be
packed in one corner of a unit square bin, leaving an $L$ shaped
area of width $\frac 14 +\gamma_1$, the opposite corner will
contain an $F$-item, and there are two additional such items next
to it on each side of the $L$ shaped area. The remaining large
$F$-items are packed into bins containing nine items each
(possibly except for one bin), such that the number of such bins
is $\lceil \frac{Y_4}9\rceil$. The total number of bins in this
packing is at most $\frac{M}5 - \frac{4Y_4}{45} +2$.
\end{proof}

The algorithm has one large $F$-item in each of the first $Y_4$
bins and therefore it uses a new bin for every item of size $\frac
34 - \gamma_1$. Thus, the total number of bins in the packing of
the algorithm (in this case) is exactly $Y_4+ \lceil \frac{M-Y_4}5
\rceil$.

We explain the continuation of the input in the case where the
continuation defined above is not used. Here, for the
construction, we state an upper bound on the number of items as
the exact number of items is not known in advance and it will be
determined during the presentation of the input. There will be at
most $1.5 M$ items of sizes slightly above $\frac 13$. We will use
the variables $S_3$ and $L_3$ to denote the numbers of items for
which condition $C_{12}$ was satisfied and was not satisfied,
respectively, in the current construction. Initialize $S_3=L_3=0$,
and increase the value of the suitable variable by $1$ when a new
item is presented.  The $i$th item of the current construction has
size $\frac 13+a_i$, and the sizes of items are in $(\frac 13,
0.33344)$. These items are called $T$-items. There is a value
$\gamma_2$ such that any item whose packing satisfies condition
$C_{12}$ has size below $\frac 13 + \gamma_2$ and any item whose
packing does not satisfy $C_{12}$ has size above $\frac 13 +
\gamma_2$. The first kind of items are called smaller $T$-items
 and the second type items are
called larger $T$-items. Present items until $8 S_3 + 15 L_3 \geq
12 M$ holds (this does not hold initially, so at least one item is
presented, and this is defined to be condition $C'$). We show that
indeed at most $1.5 M$ items are presented. If $1.5 M$ items were
already presented, $8 S_3 + 15 L_3 \geq 8 \cdot (1.5M) = 12 M$,
and therefore the construction is stopped. In what follows, let
$S_3$ and $L_3$ denote the final values of these variables. Before
the last item of this part of the input was presented, it either
was the case that $8 (S_3-1) + 15 L_3 < 12 M$ or $8 S_3 + 15 (L_3
-1) < 12 M$ (as exactly one of $S_3$ and $L_3$ was increased by
$1$ when the last item was presented), so $8 S_3 + 15 L_3 -15 < 12
M$, or alternatively, $8 S_3 + 15 L_3 \leq 12 M +15$. Moreover,
$S_3+L_3 \geq \frac {4M}5$ as $12 M \leq 8 S_3 + 15 L_3 \leq 15
(S_3+L_3)$. Let $M'=S_3+L_3$ (and we have $M'=\Theta(M)$).

Here, there are two possible continuations. The first one is
$(\lfloor \frac{M'}3 \rfloor )$ identical items, each of size
exactly $0.6$, and the second one is $\lfloor \frac{S_3}3 \rfloor$
identical items, each of size $\frac 23-\gamma_2$.

\begin{lemma}\label{lema7}
The optimal cost in the first continuation is at most $\frac M9 + \frac{7S_3}{27}+
\frac{7L_3}{27}+3 $.
\end{lemma}
\begin{proof}
A possible packing for this case consists of $\lfloor \frac{M'}3
\rfloor$ bins with one item of size $0.6$, three $T$-items, and
two $F$-items (placing the item of size $0.6$ in a corner leaves
an $L$ shaped area of width $0.4$, so we place one $T$-item in
each of the other corners and in the remaining space between each
pair of adjacent $T$-items we pack an $F$-item). As $M' \leq
\frac{3M}2$, there are $M- 2 \lfloor \frac{M'}3 \rfloor \geq 0$
unpacked $F$-items, and they are packed into exactly $\lceil
\frac{M- 2 \lfloor \frac{M'}3 \rfloor }9 \rceil \leq \frac M9 -
\frac{2M'}{27} +2$ bins, where each bin has nine items (the last
bin may have less items). In addition, there are at most two
unpacked $T$-items, and they are packed into a bin together. The
total number of bins is at most $\frac M9 + \frac{7M'}{27}+3 =
\frac M9 + \frac{7S_3}{27}+ \frac{7L_3}{27}+3 $.
\end{proof}

\begin{lemma}\label{lema8}
The optimal cost in the second continuation is at most
$\frac{S_3}3 + \frac{L_3}4+2$.
\end{lemma}
\begin{proof}
A possible packing for this case consists of $\lfloor \frac{S_3}3
\rfloor$ bins with one item of size $\frac 23 -\gamma_2$, three
small $T$-items, and two $F$-items (placing the item of size
$\frac 23 -\gamma_2$ in a corner of a unit square bin leaves an
$L$ shaped area of width $\frac 13+\gamma_2$ where the remaining
items are packed). There are at least $(S_3-2)$ $T$-items that
were packed and at least $(2\frac{S_3-2}3)$ $F$-items are packed.
There are also $\lceil \frac{L_3+2}{4} \rceil$ bins, each with at
most four $T$-items and at most five $F$-items (there is a square
with four larger items in a corner and the smaller items are
packed around them, in the $L$-shaped area of the bin). This
allows to pack the remaining $T$-items
 as there is space for at least
$S_3+L_3$ such items, and to pack all $F$-items as there is a
place for at least $2\frac{S_3-2}3+5\frac{L_3+2}4 \geq \frac{2
S_3} 3 + 5 \frac {L_3}4 \geq M$  such items, where the last
inequality holds by the condition $8 S_3 + 15 L_3 \geq 12 M$. The
total number of bins is at most $\frac{S_3}3 + \frac{L_3}4+2$.
\end{proof}

Let $Y_3$ denote the number of new bins created for the $T$-items
(where these bins were empty prior to the arrival of $T$-items).
Here, there may be previously existing bins containing larger
$T$-items (with at most four $F$-items), and $Y_3 \leq L_3$.
Consider the kinds of bins the algorithm may have after all
$T$-items have arrived. Once again, $T$-items do not necessarily
arrive, but we will deduce the numbers of different kinds of bins
the algorithm has after all $F$-items have arrived based on number
of bins existing after the arrival of $T$-items. After all
$T$-items have arrived, a non-empty bin can receive an item of
size $0.6$ if it has at most five items, out of which at most
three are $T$-items. The construction is such that any non-empty
bin except for bins with at most five $F$-items has either at
least six items in total (each of size above $\frac 14$) or it has
an item of size above $\frac 13+\gamma_2$ (or both options may
occur simultaneously), and therefore it cannot receive an item of
size above $\frac 23-\gamma_2$.

Consider a given online algorithm for {\sc SP} after the $T$-items
were presented. Let $X_{90}$ denote the number of bins with six,
seven, eight, or nine $F$-items and no $T$-items. Such a bin
cannot receive any further items in addition to its $F$-items in
any of our continuations. Let $X_{50}$ denote the number of bins
with at least one and at most five $F$-items and no $T$-items.
Such a bin can receive any item of size larger than $\frac 12$
that may arrive (but not an item of size $\frac 34 -\gamma_1$).
Let $X_{81}$ denote the number of bins with five, six, seven, or
eight $F$-items and one (small) $T$-item. Let $X_{41}$ denote the
number of bins with at least one and at most four $F$-items and
one (large) $T$-item. Let $X_{72}$ denote the number of bins with
five, six, or seven $F$-items and two (small) $T$-items. Let
$X_{42}$ denote the number of bins with
 four $F$-items and two $T$-items  (out of
which one is small and one is large). Let $X_{32}$ be the number
of bins with at least one and at most three $F$-items and two
$T$-items  (out of which one is small and one is large). Let
$X_{63}$ denote the number of bins with five or six $F$-items and
three $T$-items (all of which are small). Let $X_{43}$ denote the
number of bins with three or four $F$-items and three $T$-items
(out of which two are small and one is large). Let $X_{23}$ denote
the number of bins with one or two $F$-items and three $T$-items
(out of which two are small and one is large). Let $X_{54}$ denote
the number of bins with five $F$-items and four $T$-items (all of
which are small). Let $X_{44}$ denote the number of bins with two
or three or four $F$-items and four $T$-items (out of which three
are small and one is large). Let $X_{14}$ denote the number of
bins with one $F$-item and four $T$-items (out of which three are
small and one is large).

Let $X_{03}$ be the number of bins with no $F$-items and at least
one and at most three $T$-items, one of which is a large $T$-item,
while the others (at most two) are small. Let $X_{04}$ be the
number of bins with no $F$ items and four $T$-items, one of which
is large, while three are small.

We have three scenarios, and we use $ALG_i$ to denote the cost of
the algorithm for each one of them, in the order they were
presented. Let $A_i = \limsup _{M \rightarrow \infty}
\frac{ALG_i}{M}$. The optimal cost is always in $\Theta(M)$, and
we let $OPT_i$ denote our upper bounds on the optimal cost of the
$i$th scenario, $O_i = \liminf _{M \rightarrow \infty}
\frac{OPT_i}{M}$, and the ratio $\frac{A_i}{O_i}$ is lower bound
on the competitive ratio $R$. We use the notation
$y_i=\frac{Y_i}{M}$ and $x_{ij}=\frac{X_{ij}}M$ for those values
of $i$ and $j$ such that $Y_i$ and $X_{ij}$ are defined. Let
$\ell_3=\frac{L_3}M$ and $s_3=\frac{S_3}{M}$, so  $12 \leq 8 s_3 +
15 \ell_3 \leq 12 +\frac{15}M $, and for $M$ growing to infinity,
$8 s_3 + 15 \ell_3 = 12 $.

Let $R$ be the (asymptotic) competitive ratio. For all twenty
variables there is a non-negativity constraint. In addition, the
number of items should satisfy $\sum_{i,j} j \cdot X_{ij} \geq
S_3+L_3$ and $\sum_{i,j} i \cdot X_{ij} \geq  M$ (once again, the
first constraint is inequality and not equality as $X_{03}$ counts
also bins with less than three $T$-items, and the second
constraint is not an equality as in some cases $X_{ij}$ counts
bins with fewer than $i$ $F$-items). Using the definitions of
$Y_4$ and $Y_3$ we have $Y_4=X_{90} + X_{50} +X_{81} + X_{41} +
X_{72} + X_{42} + X_{32} +
X_{63}+X_{43}+X_{23}+X_{54}+X_{44}+X_{14}$ and
$Y_3=X_{03}+X_{04}$.

We also have $ALG_1=Y_4+\lceil \frac{M-Y_4}5 \rceil$ while $OPT_1
\leq \frac{M}5 - \frac{4Y_4}{45} +2$, so $$R \geq \frac{A_1}{O_1}
\geq \frac{1/5+4y_4/5}{1/5-4y_4/45} = \frac{9+36y_4}{9-4y_4} \ .
$$ Additionally, $ALG_2=
Y_4+Y_3-X_{50}-X_{41}-X_{32}-X_{23}-X_{03}+\lfloor \frac{M'}3
\rfloor \geq Y_4+Y_3-X_{50}-X_{41}-X_{32}-X_{23}-X_{03}+
\frac{S_3+L_3}3 -2$ while $OPT_2 \leq \frac M9+ \frac{7S_3}{27}+
\frac{7L_3}{27}+3 $, and $ALG_3=Y_4+Y_3-X_{50}+\lfloor \frac{S_3}3
\rfloor \geq Y_4+Y_3-X_{50}+\frac{S_3}3-1 $ while $OPT_3 \leq
\frac{S_3}3 + \frac{L_3}4+2$, so $$R \geq \frac{A_2}{O_2} \geq
\frac{y_4+y_3-x_{50}-x_{41}-x_{32}-x_{23}-x_{03}+ s_3/3 +
\ell_3/3}{7s_3/27+ 7\ell_3/27+1/9}$$ and $R \geq \frac{A_3}{O_3}
\geq \frac{y_4+y_3-x_{50}+s_3/3}{s_3/3 + \ell_3/4}$.

\medskip

We get the following set of constraints:
\begin{eqnarray}
& 8 s_{3}  + 15 \ell_{3} = 12 \label{eq1a}\\
& y_4=x_{90} + x_{50} +x_{81} + x_{41} + x_{72} + x_{42} + x_{32}
+ x_{63}+x_{43}+x_{23}+x_{54}+x_{44}+x_{14} \\ \label{eq2a}
& y_3=x_{03}+x_{04} \label{eq3a} \\ 
& x_{81} + x_{41}+ 2 x_{72}+ 2 x_{42} + 2 x_{32}+  3 x_{63} + 3
x_{43}+ 3 x_{23}+  4 x_{44}+ 4 x_{54} + 4 x_{14}+ 3 x_{03}+
4x_{04}
 \nonumber \\ & \geq  \ell_3 + s_3 \label{eq5a} \\
& x_{41} + x_{42} +  x_{32}+   x_{43}+  x_{23}+  x_{44}+  x_{14}+
 x_{03}+ x_{04}
= \ell_3 \label{eq6a} \\
& 9 x_{90} + 5 x_{50} + 8 x_{81} + 4 x_{41} + 7 x_{72} + 4 x_{42}
 + 3 x_{32} + 6 x_{63} + 4 x_{43} + 2 x_{23} + 5
x_{54} + 4x_{44} +
x_{14} \nonumber \\ & \geq 1  \label{eq7a} \\
& 9+36y_4 \leq R(9-4y_4)  \label{eq8a} \\
& (y_4+y_3-x_{50}-x_{41}-x_{32}-x_{23}-x_{03}+ s_3/3 + \ell_3/3)
\leq R(7s_3/27+ 7\ell_3/27+1/9)
 \label{eq9a} \\
& y_4+y_3-x_{50}+s_3/3 \leq R(s_3/3 + \ell_3/4)
 \label{eq10a}
\end{eqnarray}

The optimal objective function value of the mathematical program
of minimizing $R$ subject to all these constraints is
approximately $1.751544578513$ (and it is not smaller than this
number). Thus, we have proved $R \geq 1.751544578513$.


\section{Online class constrained bin packing (CLCBP)}\label{sec:clcbp}
In this section we exhibit our approach to proving lower bounds
for the last variant of the bin packing problem which we study
here, by improving the known lower bounds for the cases $t=2$ and
$t=3$ of {\sc CLCBP}. We will prove the following theorem.

\begin{theorem}
The competitive ratios of online algorithms for {\sc CLCBP} with
$t=2$ and $t=3$ are at least $1.717668486$ and at least
$1.80814287$, respectively.
\end{theorem}

The constructions for $t=2$ and $t=3$ have clear differences, but
the general idea is similar.  The outline of the constructions is
as follows. Start with a large number of tiny items, all of
distinct colors, so every bin of any algorithm will contain at
most $t$ tiny items. Here, the construction is such that the items
packed first into their bins are much larger than other items
(large tiny items will be larger by at least a constant
multiplicative factor than small tiny items, but they are still
very small). One option at this point is to continue with huge
items of sizes close to $1$, all of distinct colors out of the
colors of small tiny items, such that every item of size almost
$1$ can be packed into a bin with $t$ small tiny items in an
offline solution, one of which has the same color as the huge item
packed with it. Note that no large tiny item can be combined with
a huge item, so those items will be packed separately, $t$ items
per bin. The number of huge items is chosen in a way such that the
optimal cost is not increased. Another option to continue the
construction (instead of introducing the huge items) is with items
of sizes slightly above $\frac 13$, where an item packed into a
bin already containing an item of size above $\frac 13$ is smaller
than an item packed into a bin with no such item (but it could
possibly be packed with tiny items). It is ensured that bins of
the algorithm already containing $t$ (tiny) items will not be used
again by the algorithm by never introducing items of their colors
again. The sizes will be $\frac 13$ plus small values, where these
small values are much larger than sizes of tiny items (including
sizes of large tiny items). An interesting feature is that there
will be exactly {\it two} items of sizes slightly above $\frac 13$
with each color which is used for such items, where the idea is to
reuse (as much as possible) colors of tiny items packed by the
algorithm into bins with at most $t-1$ tiny items (where those
tiny items can be large or small), and never reuse colors of tiny
items packed in bins of $t$ items. In some cases (if there are too
few such colors which can be reused), new colors are used as well
for items of sizes slightly above $\frac 13$ (but there are still
two items of sizes just above $\frac 13$ for each color). After
these last items are presented, the final list of items will be
items of sizes above $\frac 12$ whose colors will match exactly
those of items of sizes in $(\frac 13,\frac 12]$ with the goal of
packing such pairs of one color together into bins of offline
solutions. There are two options for the final items. There are
either such items not much larger than $\frac 12$, or there are
items of sizes close to $\frac 23$, such that such an item having
a color of an item of size slightly above $\frac 13$ can be
combined into a bin with that item and with at most $t$  tiny
items coming from bins of the algorithm with at most $t-1$ items
(no matter whether they are small or large, but one of them has to
be of the same color). However, in the case of items of sizes
almost $\frac 23$, only small items of sizes just above $\frac 13$
will be combined with them in good offline solutions while others
are packed in pairs (of the same color whenever possible, and of
different colors otherwise, combining tiny items where possible).

First, we present the parts of the constructions that are
identical for $t=2$ and $t=3$. The condition $C_1$ will be that
the current item is not the first item of its type packed into its
bin, where a type consists of all items of similar size (the two
relevant types are tiny items and items of sizes slightly above
$\frac 13$). Let $M>1$ be a large integer divisible by $6$. The
construction starts with the first type of items, where these
items are called $E$-items or tiny items, consisting of $M$ items
constructed using Theorem \ref{cnstrct}. Let the value of $k$ be
$20$, and the resulting values $a_i$ are smaller than
${20}^{-2^{2M+2}}$. The number of tiny items presented is always
exactly $M$ (so the stopping condition is that there are $M$
items), and the size of the $i$th item is simply $a_i$. Every
$E$-item has its own color that may be reused in future parts of
the construction but not for $E$-items. Let $\eps_1$ and
$\gamma_1$ be such that the size of any $E$-item satisfying $C_1$
(which we call a small $E$-item) is below $\frac{2\eps_1}{20} <
\frac{\eps_1}t$ and the size of any $E$-item not satisfying $C_1$
(which we call a large $E$-item) is above $2\eps_1$ (but smaller
than ${20}^{-2^{2M+2}}$). Let $X_j$ (for $1 \leq j \leq t$) be the
number of bins of the algorithm with $j$ $E$-items. Let $X$ denote
the total number of bins of $E$-items, i.e., $X=\sum_{j=1}^ t
X_j$.

If huge items arrive now, their number is $\lfloor
\frac{M-X}{t} \rfloor$ and their colors are distinct colors out of
colors of small $E$-items. The size of every huge item is
$1-\eps_1$. If $X_t \leq \frac M{2t}$, there are no other
continuations. In all other cases, there are two possible
continuations except for the one with huge items, which was just discussed.

In all other continuations, items of a second type are presented
such that their number is at most $2M$, and they will be called
$T$-items. They are constructed using Theorem \ref{cnstrct} with
$k=10$, so their values of $a_i$ are in $( 10^{- 2^{2M+3}},
10^{-2^{2M+2}})$. We have (by $M \geq 1$) $\frac{10^{-
2^{2M+3}}}{{20}^{-2^{2M+2}}}=\frac{10^{2^{2M+2}}2^{2^{2M+2}}}{10^{2^{2M+3}}}=\frac{2^{2^{2M+2}}}{10^4}
> 6 > t$. The size of the $i$th $T$-item is $\frac 13+a_i$, and
here condition $C_1$ means that the $T$-item is packed by the
algorithm as the second $T$-item of its bin. Let $\eps_2$ and
$\gamma_2$ be such that a $T$-item satisfying $C_1$ (which we call
a small $T$-item) has size smaller than $\frac
13+\frac{\eps_2}{10}$ and a $T$-item not satisfying $C_1$ (which
we call a large $T$-item) has size larger than $\frac 13+\eps_2$.
The number of $T$-items is even, and their colors are such that
there are two $T$-items for each color. These colors are colors of
$E$-items that are not packed in bins of $t$ $E$-items by the
algorithm. As the number of such $E$-items is $M-t\cdot X_t$, if
the number of $T$-items is larger than $2(M-t\cdot X_t)$, new
colors (which were not used for any earlier item) are used (and
for the new colors there are also two $T$-items for each color).
The variables $Z_1$ and $Z_2$ denote the numbers of bins with at
least one $T$-item and with exactly two $T$-items, respectively,
used by the algorithm (so $Z_2 \leq Z_1$). The algorithm may use
bins with at most $(t-1)$ $E$-items to pack $T$ items (but not
bins with $(t)$ $E$-items, as no additional items have colors as
those items).

For $t=2$, the number of $T$-items is $\max\{2X_1,2X_2\}$. Since
$2X_2 \leq M$ and $2X_1 \leq 2M$, the number of $T$-items does not
exceed $2M$. For $t=3$, the stopping condition is defined as
follows. First, present items until at least one of $Z_1 + Z_2 + 6
X_3 \geq 2M - 1$, $3 Z_1 + 4 Z_2 \geq  2M - 7$ holds. Then, if the
second condition holds, stop presenting items. If the first
condition holds (and the second one does not hold), continue
presenting items until $2 Z_1 + 3 Z_2 \geq 6X_3  - 5$ holds and
stop. At this time, if the current number of $T$-items is odd, one
additional item is presented.  Thus, we guarantee that the value
of $Z_1+Z_2$ is an even number. Since the value $X_3$ is already
fixed when $T$-items are presented, we analyze the increase in the
value of each expression when a new $T$-item is presented. If a
new item is packed into a bin with no $T$-item (and it is large),
then the value of $Z_1$ increases by $1$ while the value of $Z_2$
is unchanged. Otherwise (it is small), the value of $Z_2$
increases by $1$ while the value of $Z_1$ is unchanged. Thus, the
value of $Z_1+Z_2$ can increase by at most $1$, while that of $3
Z_1 + 4 Z_2$ can increase by at most $4$, and that of $2 Z_1 + 3
Z_2$ can increase by at most $3$. Thus, there are two cases. If
the first condition that holds is $3 Z_1 + 4 Z_2 \geq 2M - 7$,
when it started to hold, the value of the left hand side was
increased by at most $4$. If another item is presented to make the
number of items even, it could increase by at most $4$ again, so
$3 Z_1 + 4 Z_2 \leq 2M$. If $Z_1 + Z_2 + 6 X_3 \geq 2M - 1$ holds
first (note that the two conditions could potentially start
holding at the same time), then still $Z_1+Z_2 +6X_3 \leq 2M$. If
in the current step it holds that $Z_1 + Z_2 + 6 X_3 \geq 2M - 1$
and  $3 Z_1 + 4 Z_2 \leq  2M - 8$, at that time, $2 Z_1 + 3 Z_2
\leq 6X_3 - 6$ holds (as otherwise, taking the sum of $Z_1 + Z_2 +
6 X_3 \geq 2M - 1$ and $2 Z_1 + 3 Z_2 \geq 6X_3  - 5$ gives $3 Z_1
+ 4 Z_2 \geq 2M - 6 > 2M-7$). Therefore in the case the first
condition holds first while the second one does not, additional
items are presented and finally $2 Z_1 + 3 Z_2 \leq  6X_3$
(counting the last two items). Thus, after all $T$-items have
arrived, it is either the case that $Z_1+Z_2  \leq 3 Z_1 + 4 Z_2
\leq  2M$ or that $Z_1 + Z_2 \leq 2 Z_1 + 3 Z_2 \leq 6X_3 \leq 2M$
(as $3X_3 \leq M$), so there are indeed at most $(2M)$ $T$-items.

A {\em matching item} for a $T$-item is an item of size above
$\frac 12$ with the same color. There are two continuations as
follows. In the first one, there are items of sizes $0.6$, such
that there is a matching item for every $T$-item (a different
matching item for every item, i.e., $Z_1+Z_2$ items of size $0.6$
in total). In the second one, there are items of sizes $\frac 23
-\frac{\eps_2}{5}$, such that every small $T$-item has a matching
item (once again, a different matching item for every item, i.e.,
$Z_2$ items in total). This concludes the description of our lower
bounds constructions for the two cases of $t=2$ and $t=3$.

\subsection{The analysis}
Let $ALG_i$ and $OPT_i$ respectively denote the costs of the
algorithm and of an optimal solution for the $i$th continuation.
We use $alg_i=\frac{ALG_i}{M}$ and $opt_i=\frac{OPT_i}{M}$. This
auxiliary notation will assist us as we would like to find the
bounds for $M$ growing to infinity. The competitive ratio
satisfies $R \geq \limsup_{M \rightarrow
\infty}\frac{alg_i}{opt_i}=\limsup_{M \rightarrow
\infty}\frac{ALG_i}{OPT_i}$. We will also use $x_i=\frac{X_i}M$
and $z_i=\frac{Z_i}M$, for values of $i$ that these variables are
defined, and $x=\frac XM$.

Consider a given online algorithm and an offline solution after
the huge items are presented.

\begin{lemma}\label{lema10}
We have $ALG_1 = X +\lfloor \frac{M-X}{t} \rfloor$ and $OPT_1 \leq
\frac Mt$.
\end{lemma}
\begin{proof}
Every huge item can be packed with $t$ small $E$-items, if one of
them has the same color as the huge item. No huge item can be
packed with a large $E$-item in one bin. Thus, the algorithm has
$\lfloor \frac{M-X}{t} \rfloor$ bins with huge items (one huge
item packed into each such bin), and all of them contain no other
items (as every bin of the algorithm with $E$-items has a large
$E$-item). A possible offline solution has $\lfloor \frac{M-X}{t}
\rfloor$ bins with a huge item and a small $E$-item of the same
color (as the color of the huge item) and $t-1$ other small
$E$-items, and there are $\frac {M-t(\lfloor \frac{M-X}{t}
\rfloor)}t=\frac Mt-\lfloor \frac{M-X}{t} \rfloor$ bins with $t$
$E$-items not packed in the previous set of bins. All $E$-items
are packed, and the total number of bins is $\frac Mt$.
\end{proof}

\begin{lemma}\label{lema11}
We have $R \geq tx+(1-x)$. If $x_t \leq \frac{1}{2t}$, then the
competitive ratio is at least $2-\frac{1}{2t}$.
\end{lemma}
\begin{proof}
In this case we consider the input without continuations. By Lemma
\ref{lema10} and by letting $M$ grow to infinity, we have $alg_1 =
\frac{t-1}t x +\frac 1t$, $opt_1 \leq \frac 1t$, and $R \geq
(t-1)x+1$. As $X_t \leq \frac{M}{2t}$, at least $M - t X_t \geq
\frac M2$ items are packed in bins containing at most $t-1$ items,
and thus $x-x_t \geq \frac {1 - t x_t}{t-1}$ and $x \geq x_t +
\frac {1 - t x_t}{t-1} = \frac{1- x_t}{t-1} \geq
\frac{1- 1/(2t)}{t-1} = \frac{2t-1}{2t(t-1)} $. 
We get $R \geq \frac{4t-1}{2t}=2 - \frac{1}{2t}$.
\end{proof}

Using the first part of the last lemma, we get $R \geq
2x+(1-x)=x+1=x_1+x_2+1$ for $t=2$, and $R \geq 2x+1$ for $t=3$. As
we prove lower bounds that are lower than $1.75$ for $t=2$ and
lower than $1.8333$ for $t=3$, by the last lemma, it is left to
deal with the case $x_t \geq \frac{1}{2t}$. Note that the
continuation of huge items is still possible for those cases. The
remaining part of the analysis is performed separately for the two
cases.

\paragraph{The case $\boldsymbol{t=2}$.} In this case we
assume $x_2 >\frac 14$ and therefore $x_1 < \frac 12 < 2x_2$. As
the number of $T$-items is $2\max\{X_1,X_2\}$, there are two
$T$-items of any color of an $E$-item packed alone in a bin by the
algorithm just after the $E$-items have arrived.

\begin{lemma}\label{lemtis2}
We have $alg_2 \geq x_2 + z_1+2z_2$, $alg_3 \geq x_2 + z_1+ z_2 $,
$opt_2 \leq z_1+ z_2$, and $opt_3 \leq
\frac{z_1+2z_2+2x_2-\max\{x_1,x_2\}}2$.
\end{lemma}
\begin{proof}
The algorithm never reuses bins with two $E$-items as no further
item has color of any of their colors. If the final items have
sizes of $0.6$, the bins with one $T$-item can possibly be reused
(but not those with two such items). The number of final items is
the same as the $T$-items, that is, $Z_1+Z_2$. If the final items
have sizes of $\frac 23 -\frac{\eps_2}{5}$, as any bin with at
least one $T$-item has a large $T$-item, no bins with $T$-items
can be reused by the algorithm. The number of final items is $Z_2$
in this case. The lower bounds on the costs of the algorithm
follow from the numbers of items of sizes above $\frac 12$ in the
final part of the input, and from the property that they cannot be
added to bins with two tiny items, to bins with two $T$-items, and
in the case of items of sizes $\frac 23 -\frac{\eps_2}{5}$ they
cannot be added to any bin with a large $T$-item (in this case
they cannot be added to any bin with at least one $T$-item).

Consider the following offline solutions. If the final items have
sizes of $0.6$, every bin contains a $T$-item and its matching
item of size $0.6$. It also contains an $E$-item of the same
color, if it exists (it is also possible that it exists but it is
packed in another bin with a $T$-item of the same color), and at
most one $E$-item of another color. As $Z_1+Z_2 \geq 2X_1$ and
$Z_1+Z_2 \geq 2X_2$, every $E$-item packed alone in the algorithm
(after all the $E$-items arrive) has a $T$-item of its color
(there are two items with this color, and it can be packed with
one of them). Given the number of bins of this solution, it is
possible to add (at most) one $E$-item, which is packed in bins of
two $E$-items by the algorithm, to each bin containing an item of
size $0.6$ (as the number of such $E$-items is $2X_2$ and the
number of bins is the number of $T$-items, that is, at least
$2X_2$). The total size of items in every bin is below $0.94$.
Thus, those $Z_1+Z_2$ bins are packed in a valid way and contain
all items.

If the final items have sizes of $\frac 23 -\frac{\eps_2}{5}$, as
$E$-items have sizes no larger than $\frac{\eps_2}{60}$, it is
possible to pack one small $T$-item with its matching item of size
$\frac 23 -\frac{\eps_2}{5}$, and at most two $E$-items, one of
which has the same color as the small $T$-item. As there are $Z_2$
small $T$-items, there are at least $Z_1-Z_2$ large $T$-items such
that the other $T$-item of the same color is large, and therefore
there are at least $\frac{Z_1-Z_2}2$ pairs of large $T$-items with
common colors (as $Z_1+Z_2$ is even, $Z_1-Z_2$ is even too). There
are $Z_1-Z_2$ large $T$-items that are packed in pairs, such that
$\frac{Z_1-Z_2}2$ pairs of two large $T$-items of the same color
are packed together with one $E$-item of their color and one
$E$-item of another color (because it cannot contain items of an
additional color). Note that even if there is a larger number of
pairs of large $T$-items with common colors, exactly
$\frac{Z_1-Z_2}2$ pairs are packed in this way. The other large
$T$-items and unpacked $E$-items are simply packed in pairs. Note
that there are $(2X_2)$ $E$-items with unique colors (where no
other item has the same color). We have packed
$Z_2+\frac{Z_1-Z_2}2 =\frac{Z_1+Z_2}2 = \max\{X_1,X_2\} < 2X_2$
items (recall that the number of $T$-items is $2\cdot
\max\{X_1,X_2\}$ and it is also $Z_1+Z_2$, while there are
$(2X_2)$ $E$-items of unique colors and the number of other
$E$-items is $X_1 \leq \max\{X_1,X_2\}$, while the number of
colors of $T$-items is $\max\{X_1,X_2\}$) that are $E$-items with
unique colors, so there are still such items to be packed. There
are $Z_1$ large $T$-items, and therefore $Z_2$ such items remain.
Therefore, as the number of unpacked $E$-items of unique colors is
$2X_2-\frac{Z_1+Z_2}2$, an additional $\lceil
\frac{2X_2-\frac{Z_1+Z_2}2+Z_2}2 \rceil=X_2 - \lfloor
\frac{Z_1-Z_2}4 \rfloor$ bins are used for the pairs. Thus, the
number of bins is at most $Z_2 + \frac{Z_1-Z_2}2 +
X_2-\frac{Z_1-Z_2}4+1=X_2+\frac{Z_1+3Z_2}4+1$. In the case
$Z_1+Z_2=2X_2$, we have $X_2+\frac{Z_1+3Z_2}4=
\frac{Z_1}2+Z_2+\frac{X_2}2$, and in the case $Z_1+Z_2=2X_1$, we
have $X_2+\frac{Z_1+3Z_2}4= \frac {Z_1}2+Z_2+X_2-\frac{X_1}2$. In
both cases the number of bins is at most $\frac
{Z_1}2+Z_2+X_2-\frac{\max\{X_1,X_2\}}2$. The other $E$-items are
packed with $T$-items of their colors.
\end{proof}

Here we solve two mathematical programs, both minimizing $R$ under
constraints including non-negativity constraints for all
variables, and the properties $x_1 \leq 2x_2$, $x_1+x_2 +1 \leq R$, $ x_2 + z_1 + 2 z_2  \leq R(z_1
+ z_2) $, $z_2 \leq z_1$, and $x_1 + 2 x_2 = 1$.

The first program is for the case $x_2 \geq x_1$, which is one of
the constraints (where $z_1+ z_2 = 2x_2$). The other constraints
are  $z_2 +
z_1 - 2 x_2 = 0$, $ x_2 +  z_1 +  z_2 \leq R(z_1/2 +  z_2 + x_2/2
)$. Solving the program shows that $R \geq 1.7320507$ in this case.

The second program is for the case $x_1 \geq x_2$, which is one of
the constraints (and here $z_1+z_2 = 2x_1$). The other
constraints are  $z_2 + z_1 - 2 x_1 = 0$, $ x_2 + z_1 + z_2 \leq
R(z_1/2 + z_2  + x_2 - x_1/2 )$. Solving the program shows that $R \geq 1.717668486$ in this case.

\paragraph{The case $\boldsymbol{t=3}$.} In this case we
assume $x_3 >\frac 16$.

\begin{lemma}\label{lemtis3}
We have $alg_2 \geq x_3 + z_1+2z_2$, $alg_3 \geq x_3 + z_1+ z_2 $,
$opt_2 \leq z_1+ z_2$, and $opt_3 \leq
\frac{z_1+2z_2}2$.
\end{lemma}
\begin{proof}
The algorithm never reuses bins with three $E$-items as no further
item has any color of their colors. Other than that, the arguments
for the costs of the algorithm are the same as in the case $t=2$.

Next, we analyze offline solutions. In both cases of final items,
the difference with the case $t=2$ is that every bin can contain
two $E$-items whose colors are unique (either because they come
from bins with three $E$-items of the algorithm or because the
number of colors of $T$-items is smaller than the number of items
coming from bins of the algorithm with less than three items). It
is possible to add such items to the bins as the size of three
$E$-items is still below $\frac{\eps_2}{10}$.

We first calculate the number of $E$-items of unique colors (that
is, $E$-items of colors that appear only once for the entire
input).  In the case where $Z_1+Z_2+6X_3 \geq 2M-1$ we have in
fact $Z_1+Z_2+6X_3 \geq 2M$ as the value $Z_1+Z_2$ is even. In
this case every $E$-item packed in a bin with less than three
$E$-items by the algorithm has two $T$-items of its color, and it
can always be packed with one of them. In this case the number of
$E$-items of unique colors is $3X_3$. Otherwise, the number of
$E$-items of unique colors is $M-\frac{Z_1+Z_2}2$, as there are
$(Z_1+Z_2)$ $T$-items, and there are two $T$-items of each color.

We claim that in the case of final items of sizes $\frac 23
-\frac{\eps_2}{5}$, it is possible to pack all $E$-items of unique
colors, possibly except for a constant number of items which can
be packed separately into a constant number of bins. We claim that
there is always space for at least $(Z_1+1.5Z_2-1)$ $E$-items of
unique colors. The difference with the case $t=2$ is that the bins
with the final items can receive two $E$-items of unique colors
and not only one (and there are $Z_2$ such bins). The bins with
pairs of large $T$-items of one color can receive two $E$-items of
unique colors (and there are $\frac{Z_1-Z_2}{2}$ such bins), and
the remaining bins, with two large $T$-items of distinct colors
can receive one such $E$-item (and there are $\lceil \frac{Z_2}{2}
\rceil$ such bins). Thus, it is possible to pack at least $(2Z_2+2
\frac{Z_1-Z_2}2 +\lceil \frac {Z_2}2 \rceil)$ $E$-items of unique
colors. If their number if $3X_3$, we also have $Z_1 +1.5 Z_2 \geq
3X_3 - 2.5$, so excluding a constant number of such items, all of
them are packed. If their number is $M-\frac{Z_1+Z_2}2$, we also
have $3Z_1+4Z_2 \geq 2M-7$, so $M-\frac{Z_1+Z_2}2 \leq Z_1+1.5Z_2
+3.5$. Thus, we find $opt_2 \leq \frac{z_1+2z_2}2$. In the case
where the final items have sizes of $0.6$, it is possible to pack
$(2Z_1+2Z_2)$ $E$-items of unique colors in those bins, and $opt_2
\leq z_1+ z_2$.
\end{proof}

Here we also solve two mathematical programs, both minimizing $R$ under
constraints including non-negativity constraints for all
variables. Other constraints are $x_1 + 2 x_2 + 3 x_3 = 1$, $x=x_1+x_2+x_3$, $1+2x \leq R$, $z_2 \leq z_1$, $x_3 + z_1 + z_2  \leq R(z_1+2z_2)/2$, and
$x_3 + z_1 + 2 z_2 \leq R (z_1+z_2)$.

The first program is for the case where  $Z_1 + Z_2 + 6 X_3 \geq
2M - 1$ and $-5 \leq 2 Z1 + 3 Z2 - 6X3 \leq 0$. These properties
result in the constraints $z_1 + z_2 + 6 x_3 \geq  2$ and $2 z_1 +
3 z_2 -6 x_3 = 0$. Solving the program shows that $R \geq
1.902018$ in this case.

The second program is for the case where $Z_1 + Z_2 + 6 X_3 \leq
2M + 12$ and $2M -7  \leq 3 Z_1 + 4 Z_2 \leq  2M $ hold. Note that
if we stop presenting $T$-items due to the second case where $2M
-7  \leq 3 Z_1 + 4 Z_2$ it means that in the previous
(even-indexed) step the first condition  $Z_1 + Z_2 + 6 X_3 \geq
2M - 1$ did not hold.  Therefore, at that time $Z_1 + Z_2 + 6 X_3
\leq 2M -2$ holds, and the value of the left hand side may
increase by at most $7$ in one step (and thus by at most $14$ in
the last two steps). Those properties result in the constraints
$z_1 + z_2 + 6 x_3 \leq 2$ and $3 z_1 + 4 z_2 = 2$. Solving the
program shows that $R \geq 1.80814287$ in this case.

\section{Summary}
We showed that the method of designing fully adaptive instances,
previously used for cardinality constrained bin packing and vector
packing \cite{BBDEL} (see also \cite{Blitz,BCKK04,FK13}) can be
used to improve the known lower bounds for several additional bin
packing problems. We analyzed its effect (together with many
additional ideas) for several variants, and expect that it could
be useful for a number of other variants as well.

\bibliographystyle{abbrv}

\end{document}